\documentclass[11pt]{article}
\usepackage{mathrsfs}
\usepackage{amsfonts}
\usepackage{amsmath}
\usepackage{graphicx}
\usepackage{texdraw}
\usepackage{rotating}
\usepackage[round]{natbib}
\begin{document}
\sloppy
\newenvironment{proof}%
{\begin{trivlist}\item[\hspace*{\labelsep}{\it Proof.\/}]}%
{\hfill$\Box$\end{trivlist}}
\newtheorem{theo}{Theorem}
\newtheorem{lemma}[theo]{Lemma}
\newtheorem{prop}[theo]{Proposition}
\newtheorem{coro}[theo]{Corollary}
\newtheorem{obse}[theo]{Observation}
\newtheorem{claim}[theo]{Claim}
\newtheorem{rema}[theo]{Remark}
\newtheorem{defi}[theo]{Definition}
\newtheorem{conj}[theo]{Conjecture}
\newcommand{\qed}{\hfill\rule{1.8mm}{1.8mm}}
\newcommand{\rz}{\mbox{{\rm I$\!$R}}}
\newcommand{\nz}{\mbox{{\rm I$\!$N}}}
\newcommand{\eps}{\varepsilon}
\newcommand{\al}{\alpha}
\newcommand{\seq}[1]{\langle #1\rangle}
\newcommand{\flo}[1]{\lfloor #1 \rfloor}
\newcommand{\cei}[1]{\lceil  #1 \rceil}
\newcommand{\np}{\mbox{${\cal N\cal P}$}}
\newcommand{\p}{\mbox{${\cal P}$}}
\newcommand{\para}{\medskip\noindent}
\newcommand{\bbb}{{\cal B}}
\newcommand{\lll}{{\cal L}}
\newcommand{\iii}{{\cal I}}
\newcommand{\opt}{\mbox{\rm OPT}}
\newcommand{\head}[1]
{\markright{\hbox to 0pt{\vtop to 0pt{\hbox{}\vskip 3mm \hrule width
\textwidth \vss} \hss}{\sc #1}}}
\title{Simultaneous approximation for scheduling problems}

\author{Long Wan\thanks{cocu3328@163.com. Department of Mathematics, Zhejiang University, Hangzhou, 310027, China.}}
\date{}
\maketitle \baselineskip 17pt
\begin{abstract}
Motivated by the problem to approximate all feasible schedules by
one schedule in a given scheduling environment, we introduce in
this paper the concepts of strong simultaneous approximation ratio
(SAR) and weak simultaneous approximation ratio (WAR). Then we
study the two parameters under various scheduling environments,
such as, non-preemptive, preemptive or fractional scheduling on
identical, related or unrelated machines.

\vskip 2mm\noindent{\bf Keywords.} {scheduling; simultaneous
approximation ratio; global fairness}
\end{abstract}

\section{Introduction}
In the scheduling research, people always hope to find a schedule
which achieves the balance of the loads of the machines well. To
the end, some objective functions, such as minimizing makespan and
maximizing machine cover, are designed   to find a reasonable
schedule. Representative publications can be found in
\citet{GR66}, \citet{GR69}, \citet{DE82}, and \citet{CS92} among
many others. But these objectives don't reveal the global fairness
for the loads of all machines. Motivated by the problem to
approximate all feasible schedules by one schedule in a given
scheduling environment and so realizing the global fairness, we
present two new parameters: strong simultaneous approximation
ratio (SAR) and weak simultaneous approximation ratio (WAR).

Our research is also enlightened from the research on global
approximation of vector sets. Related work  can be found in
\citet{BH01}, \citet{GO01}, \citet{GO05}, \citet{KL01} and
\citet{KU06}. \citet{KL01} proposed the notion of the
coordinate-wise approximation for the fair vectors of allocations.
Based on this notion, \citet{KU06} introduced the definitions of
the global approximation ratio and the global approximation ratio
under prefix sums.

For a given instance $\mathcal{I}$ of a   minimization
  problem, we use $V(\mathcal{I})$ to denote the set of vectors
induced by all feasible solutions of $\mathcal{I}$. For a vector
$X= (X_1, X_2, \cdots, X_m)\in V(\mathcal{I})$, we use
$\overleftarrow{X}$ to denote the vector in which the coordinates
(components) of $X$ are sorted in non-increasing order, that is,
$\overleftarrow{X} =(X'_1, X'_2, \cdots, X'_m)$ is a resorting of
$(X_1, X_2, \cdots, X_m)$ so that $X'_1\geq X'_2\geq \cdots \geq
X'_m$. For two vectors $X,Y\in V(\mathcal{I})$, we write
$X\preceq_{c}Y$ if $X_i\preceq Y_i$ for all $i$. The global
approximation ratio of a vector $X \in V(\mathcal{I})$, denoted by
$c(X)$, is defined to be the infimum of $\alpha$ such that
$\overleftarrow{X}\preceq_{c}\alpha\overleftarrow{Y}$ for all
$Y\in V(\mathcal{I})$. Then the best global approximation ratio of
  instance $\mathcal{I}$ is defined to be
$c^*(\mathcal{I})=\inf_{X\in V(\mathcal{I})}c(X)$.  For a vector
$X\in V(\mathcal{I})$, we use $\sigma(X)$ to denote the vector in
which the $i$-th coordinate is equal to the sum of the first $i$
coordinates of $X$. We write $X\preceq_{s}Y$ if
$\sigma(\overleftarrow{X})\preceq_{c}\sigma(\overleftarrow{Y})$.
The global approximation ratio under prefix sums of a vector $X
\in V(\mathcal{I})$, denoted by $s(X)$, is defined to be the
infimum of $\alpha$ such that $X\preceq_{s}\alpha Y$ for all $Y\in
V(\mathcal{I})$. Then the best global approximation ratio under
prefix sums of instance $\mathcal{I}$ is defined to be
$s^*(\mathcal{I})=\inf_{X\in V(\mathcal{I})}s(X)$.

In the terms of scheduling, the above concepts about the global
approximation of vector sets can be naturally formulated as  the
 simultaneous approximation of scheduling problems. Let
$\mathcal{I}$ be an instance of a scheduling problem ${\cal P}$ on
$m$ machines $M_1, M_2, \cdots, M_m$, and let ${\cal S}$ be the
set of all feasible schedules of $\mathcal{I}$. For a feasible
schedule $S \in {\cal S}$, the \emph{load} $L^S_i$ of machine
$M_i$ is defined to be the time by which the machine finishes all
the process of the jobs and the parts of the jobs assigned to it.
The  $L(S)=(L^S_1, L^S_2, \cdots, L^S_m)$ is called the \emph{load
vector} of machines under $S$. Then $V(\mathcal{I})$ is defined to
be the set of all load vectors of instance $\mathcal{I}$. We write
$c(S) = c(L(S))$ and $s(S)= s(L(S))$ for each $S \in {\cal S}$.
Then $c^*(\mathcal{I})=\inf_{S\in {\cal S}} c(S)$ and
$s^*(\mathcal{I})=\inf_{S\in {\cal S}} s(S)$. The \emph{strong
simultaneous approximation ratio} of problem ${\cal P}$ is defined
to be $SAR({\cal P})=\sup_{\mathcal{I}}c^*(\mathcal{I})$, and the
\emph{weak simultaneous approximation ratio} of problem ${\cal P}$
is defined to be $WAR({\cal
P})=\sup_{\mathcal{I}}s^*(\mathcal{I})$.

A scheduling problem is usually characterized by the machine type
and the job processing mode. In this paper, the machine types under
consideration are identical machines, related machines and unrelated
machines, and the job processing modes under consideration are
non-preemptive, preemptive and fractional. Let
$\mathcal{J}=\{J_1,J_2,\cdots,J_n\}$ and
$\mathcal{M}=\{M_1,M_2,\cdots,M_m\}$ be the set of jobs and the set
of machines, respectively. The processing time of $J_j$ on $M_i$ is
$p_{ij}$. If $p_{ij}=p_{kj}$ for $i\neq k$, the machine type is
\emph{identical machines}. In this case $p_j$ is used to denote the
processing time of $J_j$. If $p_{ij}=\frac{p_j}{s_i}$ for all $i$,
the machine type is \emph{related machines}. In this case, $p_j$ is
called the standard processing time of $J_j$ and $s_i$ is called the
processing speed of $M_i$. If there is no restriction for $p_{ij}$,
the machine type is \emph{unrelated machines}. If each job must be
non-preemptively processed on some machine, the processing mode is
\emph{non-preemptive}. If each job can be processed preemptively and
can be processed on  at most one machine at any time, the processing
mode is \emph{preemptive}. If each job can be partitioned into
different parts which can be processed on different machines
concurrently, the processing mode is  \emph{fractional}. Each
machine can process at most one job at any time under any processing
mode.

Since we cannot avoid the worst schedule in which all jobs are
processed on a common machine, it can be easily verified that, under
each processing mode, $SAR({\cal P})= m$ for identical machines,
$SAR({\cal P})= (s_1+s_2+\cdots+s_m)/s_1$ for related machines with
speeds $s_1\geq s_2\geq\cdots\geq s_m$, and $SAR({\cal P})=+\infty$
for unrelated machines.

We then concentrate our research on the weak simultaneous
approximation ratio $WAR({\cal P})$ of the scheduling problems
defined above.  The main results  are demonstrated in table
\ref{table:1}.
\begin{table}[!hft]\label{table:1}
\centering
\begin{tabular}{|c|c|c|c|}
\hline &identical machines&related machines&unrelated machines\\
\cline{1-4} non-preemptive
processing&$1<{WAR}\leq\frac{3}{2}$&$\frac{\sqrt{m}+1}{2}\leq{WAR}\leq
\sqrt{m}$&$\frac{\sqrt{m}+1}{2}\leq{WAR}\leq\sqrt{m}$\\ \cline{1-4}
preemptive processing&$1$&$\frac{\sqrt{m}+1}{2}\leq{WAR}\leq
\sqrt{m}$&$\frac{\sqrt{m}+1}{2}\leq{WAR}\leq\sqrt{m}$\\ \cline{1-4}
fractional processing&$1$&$\frac{\sqrt{m}+1}{2}$&$\frac{\sqrt{m}+1}{2}\leq{WAR}\leq\sqrt{m}$\\
  \hline
\end{tabular}
\vspace{-0.3cm} \caption{The weak simultaneous approximation ratio
of various scheduling problems}
\end{table}

For convenience, we use  $P$, $Q$ and $R$ to represent identical
machines, related machines and unrelated machines, respectively,
and use $NP$, $PP$ and $FP$ to represent non-preemptive,
preemptive and fractional processing, respectively. Then the
notation $Pm(NP)$ represents the scheduling problem on $m$
identical machines under non-preemptive processing mode. Other
notations for scheduling problems can be similarly understood.

This paper is organizes as follows. In Section 2, we study the
weak simultaneous approximation ratio for scheduling on identical
machines. In Section 3, we study the weak simultaneous
approximation ratio for scheduling on related machines. In Section
4, we study the weak simultaneous approximation ratio for
scheduling on unrelated machines.

\section{Identical machines}

For problem $P2(NP)$, we have $s(S)=1$ for every schedule $S$ which
minimizes the makespan. So $WAR(P2(NP)) =1$. For problem $Pm(NP)$
with $m\geq 3$, the following instance shows that $WAR(Pm(NP)) >1$.
In the instance, there are $m$ jobs with processing time $m-1$,
$(m-1)(m-2)$ jobs with processing time $m$ and a big job with
processing time $(m-1)^2+r_m$, where
$r_m=\frac{\sqrt{(m^3-m^2-m-2)^2+4m(m-1)(m-2)}-(m^3-m^2-m-2)}{2}$.
It can be verified that $0<r_m<m-2$. Let $S$ be the schedule in
which the $m$ jobs with processing time $m-1$ are scheduled on one
machine, the big job with with processing time $(m-1)^2+r_m$ is
scheduled on one machine,  and the remaining $(m-1)(m-2)$ jobs with
processing time $m$ are scheduled on the remaining $m-2$ machines
averagely. Let $T$ be the schedule in which the big job is scheduled
on one machine together with a job of processing time $m-1$, and
each of the remaining machines has a job of processing time $m-1$
and $m-2$ jobs of processing time $m$. Then the makespan of schedule
$S$ is $m(m-1)$ and the $(m-1)$-th prefix sum of
$\overleftarrow{L(T)}$ is $m(m-1)^2-(m-2-r_m)$. Now consider an
arbitrary schedule $\varrho$. If the big job is scheduled on one
machine solely, then the $(m-1)$-th prefix sum of
$\overleftarrow{L(R)}$ is at least $m(m-1)^2$. Thus, by considering
the $(m-1)$-th prefix sums of $\overleftarrow{L(T)}$ and
$\overleftarrow{L(R)}$, we have $s(R) \geq
\frac{m(m-1)^2}{m(m-1)^2-(m-2-r_m)}=1+\frac{r_m}{m(m-1)}$. If the
big job is scheduled on one machine together with at least one other
job, then the makespan of schedule $R$ is at least
$(m-1)+(m-1)^2+r_m$. Thus, by considering the makespans of $S$ and
$R$, we have $s(R) \geq  1+\frac{r_m}{m(m-1)}$. It follows that
$WAR(Pm(NP)) \geq  1+\frac{r_m}{m(m-1)} >1$ for
 $m\geq3$.

To establish the upper of $WAR(Pm(NP))$, we first present a simple
but useful lemma.

\begin{lemma}\label{le:2}
Let $X,Y$ be two vectors of $n$-dimension and let $X',Y'$ be two
vectors of two-dimension. If $X\preceq_{s}Y$ and
$X'\preceq_{s}Y'$, then $(X, X')\preceq_{s}(Y,Y')$.
\end{lemma}
\begin{proof}
Suppose that $X'=(x_1, x_2)$ and $Y'=(y_1, y_2)$. Without loss of
generality, we may further assume that $x_1\geq x_2$ and $y_1\geq
y_2$. Then $x_1\leq y_1$ and $x_1+x_2 \leq y_1+y_2$. Let $Z_x =(X,
X')$ and $Z_y=(Y, Y')$. For $Z\in \{Z_x, Z_y\}$, we use
$(\overleftarrow{Z})_k$ to denote the $k$-th coordinate of
$\overleftarrow{Z}$, and use $|\overleftarrow{Z}|_k$ to denote the
sum of the first $k$ coordinates of $\overleftarrow{Z}$ for $1\leq
k \leq n+2$. Similar notations are also used for $X$ and $Y$.
Given an index $k$ with $1\leq k \leq n+2$, we use $\delta (k,
X')$ to denote the number of elements in $\{x_1, x_2\}$ included
in the first $k$ coordinates of $\overleftarrow{Z_x}$, and $\delta
(k, Y')$  the number of elements in $\{y_1, y_2\}$ included in the
first $k$ coordinates of $\overleftarrow{Z_y}$. Then $0\leq \delta
(k, X'), \delta (k, Y')\leq 2$.

If $\delta (k, X')=\delta (k, Y')$, then we clearly have
$|\overleftarrow{Z_x}|_k \leq |\overleftarrow{Z_y}|_k$.

If $\delta (k, X')=0$, then $|\overleftarrow{Z_x}|_k =
|\overleftarrow{X}|_k \leq |\overleftarrow{Y}|_k \leq
|\overleftarrow{Z_y}|_k$.

If $\delta (k, Y')=0$ and $\delta (k, X') \geq 1$, we suppose that
$x_1$ is the $i$-th coordinate of $\overleftarrow{Z_x}$. Then, for
each $j$ with $i\leq j\leq k$,  $(\overleftarrow{Z_x})_j \leq x_1
\leq y_1 \leq (\overleftarrow{Z_y})_j$. Consequently,
$|\overleftarrow{Z_x}|_k = |\overleftarrow{X}|_{i-1} +\sum_{i\leq
j\leq k} (\overleftarrow{Z_x})_j \leq |\overleftarrow{Y}|_{i-1}
+\sum_{i\leq j\leq k} (\overleftarrow{Z_y})_j =
|\overleftarrow{Z_y}|_k$.

If $\delta (k, X')=2$ and $\delta (k, Y')=1$, then
$(\overleftarrow{Y})_{k-1} \geq y_2$. Thus,
$|\overleftarrow{Z_x}|_k = |\overleftarrow{X}|_{k-2} +x_1 +x_2
\leq |\overleftarrow{Y}|_{k-2} +y_1+y_2 \leq
|\overleftarrow{Y}|_{k-1} +y_1  = |\overleftarrow{Z_y}|_k$.

If $\delta (k, X')=1$ and $\delta (k, Y')=2$, then
$(\overleftarrow{Y})_{k-1} \leq y_2$. Thus,
$|\overleftarrow{Z_x}|_k = |\overleftarrow{X}|_{k-1} +x_1
  \leq |\overleftarrow{Y}|_{k-1} +y_1  \leq  |\overleftarrow{Y}|_{k-2}
  +y_1 + y_2  =
|\overleftarrow{Z_y}|_k$.

The above discussion covers all possibilities. Then the lemma
follows.
\end{proof}

\begin{theo}
$WAR(Pm(NP)) \leq \frac{3}{2}$ for $m\geq 4$ and $WAR(P3(NP)) \leq
 \sqrt{5}-1\approx 1.236$.
\end{theo}
\begin{proof}
Consider an instance of $n$ jobs on $m\geq 4$ identical machines
with ${\cal J}=\{J_1,J_2,\cdots,J_n\}$ and ${\cal
M}=\{M_1,M_2,\cdots, M_m\}$. We assume that $p_1\geq p_2\geq
\cdots\geq p_n$. Let $S$ be a schedule produced by LPT algorithm
(which is the LS algorithm with the jobs being given in the LPT
order) such that $L^S_1\geq L^S_2\geq \cdots\geq L^S_m$. Then
$L(S)=\overleftarrow{L(S)}=(L^S_1, L^S_2, \cdots, L^S_m)$. If
$n\leq m$, it is easy to verify that $s(S)=1$. Hence we assume in
the following that $n\geq m+1$. Then some machine has at least two
jobs in $S$.

Let $i_0$ be the smallest index such that either $M_{i_0+1}$ has at
least three jobs in $S$, or $M_{i_0+1}$ has exactly two jobs in $S$
and the size of the shorter job on $M_{i_0+1}$
  is at most half of the size of the longer job on $M_{i_0+1}$.
If there is no such index, we set $i_0=m$. Then $i_0\geq 0$, and in
the case $i_0\geq 1$, each of $M_1,M_2,\cdots,M_{i_0}$ has at most
two jobs in $S$.
 Let $J_k$ be the shortest job scheduled on
$M_1,M_2,\cdots,M_{i_0}$ and set ${\cal
J}_k=\{J_1,J_2,\cdots,J_k\}$. Then ${\cal J}_k$ contains the jobs
scheduled on $M_1,M_2,\cdots,M_{i_0}$. We use $M_{k'}$ to denote
the machine occupied by $J_k$ in $S$. Let $T$ be the schedule
derived from $S$ by deleting $J_{k+1},J_{k+2},\cdots,J_n$. Then
$T$ is an LPT-schedule for ${\cal J}_k$ with
$L^T_i=L^S_i,i=1,2,\cdots,i_0$. We claim that $s(T)=1$. In the
case $i_0=0$, the claim holds trivially. Hence, we assume in the
following that $i_0\geq 1$.

If  each of $M_1,M_2,\cdots,M_{i_0}$ has only one job in $S$, then
$i_0=k\leq m$ and it is easy to see that $s(T)=1$.

Suppose in the following that at least one of
$M_1,M_2,\cdots,M_{i_0}$ has exactly two jobs in $S$. Then $m+1\leq
k \leq 2m$ and the machine $M_{k'}$ has exactly two jobs, say $J_t$
and $J_k$, in $S$. Note that there are at most two jobs on each
machine in $T$. (Otherwise, some machine $M_i$ with $i\geq i_0+1$
has $r\geq 3$ jobs, say $J_{h_1},J_{h_2},\cdots,J_{h_r}$, in $T$. By
LPT algorithm, $p_t\geq\sum^{r-1}_{j=1}p_{h_j}\geq 2p_k$,
contradicting the choice of $i_0$.)  From the LPT algorithm, we have
$t= {2m+1-k}$. By the choice of $i_0$, we have
$p_k>\frac{1}{2}p_{2m+1-k}$.

Let $R$ be an arbitrary schedule for ${\cal J}_k$. If each machine
has at most two jobs in $R$, we set $R_1=R$.  If some machine $M_x$
has at least three jobs in $R$, by the pigeonhole principle, a
certain machine $M_y$ has either no job or exactly one job in
$\{J_{2m+1-k},J_{2m+2-k},\cdots,J_k\}$. Let $R'$ be the schedule
obtained from $R$ by moving the shortest job, say $J_{x'}$, on $M_x$
to $M_y$. Then $L^{R'}_x \geq 2p_k>p_{2m+1-k}\geq L^R_y$ and
$L^{R'}_y= L^{R}_y+ p_{x'}\geq L^R_y$. Note that $L^{R}_x \geq
L^{R'}_x , L^{R'}_y \geq L^{R}_y$ and $L^{R}_x + L^{R}_y = L^{R'}_x
+L^{R'}_y$. Then we have $L(R')\preceq_{s}L(R)$ by lemma \ref{le:2}.
This procedure is repeated until we obtain a schedule $R_1$ so that
each machine has at most two jobs in $R_1$. Then we have
$L(R_1)\preceq_{s}L(R)$.

If $J_1, J_2, \cdots, J_m$ are processed on distinct machines,
respectively, in $R_1$, we set $R_2= R_1$. If some machine $M_x$
has two jobs $J_{x'}, J_{x''}\in \{J_1, J_2, \cdots, J_m\}$ in
$R_1$, by the pigeonhole principle, a certain machine $M_y$ is
occupied by at most two jobs in $\{J_m, J_{m+1}, \cdots, J_k\}$.
Suppose that $p_{x'}\geq p_{x''}$ and $J_{y'}$ is the shorter job
on $M_y$. Let $R'_1$ be the schedule obtained from $R_1$ by
shifting $J_{x''}$ to $M_y$ and shifting $J_{y'}$ to $M_x$. Then
$L^{R_1}_x \geq L^{R'_1}_x , L^{R'_1}_y \geq L^{R_1}_y$ and
$L^{R_1}_x + L^{R_1}_y = L^{R'_1}_x +L^{R'_1}_y$. Consequently, by
lemma \ref{le:2}, $L(R'_1)\preceq_{s}L(R_1)$. This procedure is
repeated until we obtain a schedule $R_2$ so that $J_1, J_2,
\cdots, J_m$ are processed on distinct machines, respectively, in
$R_2$. Then we have $L(R_2)\preceq_{s}L(R_1)$.

Without loss of generality, we assume that $J_j$ is processed on
$M_j$ in $R_2$, $1\leq j\leq m$. Let $t= k-m$. Then the $t$ jobs
$J_{m+1}, J_{m+2}, \cdots, J_k$ are processed on $t$ distinct
machines in $R_2$. For convenience, we add another $m-t$ dummy jobs
with sizes 0 in $R_2$ so that each machine has exactly two jobs. We
define a sequence of $t$ schedules $R_2^{(1)}, R_2^{(2)}, \cdots,
R_2^{(t)}$ for ${\cal J}_k$ by the following way.

Initially we set $R_2^{(0)}=R_2$. For each $i$ from 1 to $t$, the
schedule $R_2^{(i)}$ is obtained from $R_2^{(i-1)}$ by exchanging
the shorter job on $M_{m-i+1}$ with job $J_{m+i}$.

We only need to show that $L(R_2^{(i)})\preceq_{s}L(R_2^{(i-1)})$
for each $i$ with $1\leq i\leq t$. Note that the jobs $J_{m+1},
J_{m+2}, \cdots, J_{m+i-1}$ are processed on machines $M_m, M_{m-1},
\cdots, M_{m-i+2}$, respectively, in $R_2^{(i-1)}$. If $J_{m+i}$ is
processed on $M_{m-i+1}$ in $R_2^{(i-1)}$, we have $R_2^{(i)}
=R_2^{(i-1)}$ and so $L(R_2^{(i)})\preceq_{s}L(R_2^{(i-1)})$. Thus
we may assume that $J_{m+i}$ is processed on a machine $M_x$ with
$x\leq {m-i}$ in $R_2^{(i-1)}$. Let $J_{j}$ be the shorter job on
$M_{m-i+1}$ in $R_2^{(i-1)}$. Then $p_j \leq p_{m+i}$ and $p_x \geq
p_{m-i+1}$. It is easy to see that $(L^{R_2^{(i)}} _x,
L^{R_2^{(i)}}_{m-i+1}) = (p_x + p_j, p_{m-i+1} +p_{m+i}) \preceq_{s}
(p_x + p_{m+i}, p_{m-i+1} +p_j) = (L^{R_2^{(i-1)}} _x,
L^{R_2^{(i-1)}}_{m-i+1})$. Consequently, by lemma \ref{le:2},
$L(R_2^{(i)})\preceq_{s}L(R_2^{(i-1)})$.

The above discussion means that $L(R_2^{(t)}) \preceq_{s} L(R_2)
\preceq_{s} L(R_1) \preceq_{s} L(R)$. Since $R_2^{(t)}$ is
essentially an LPT-schedule, we have $\overleftarrow{L(T)} =
\overleftarrow{L(R_2^{(t)})}$, and so, $L(T) \preceq_{s}
L(R_2^{(t)})$. It follows that $L(T) \preceq_{s} L(R)$. The claim
follows.

Now let $\bar{S}$ be an arbitrary schedule for ${\cal J}$, and let
$\bar{T}$ be the   schedule for ${\cal J}_k$ derived from $\bar{S}$
by deleting jobs $J_{k+1},J_{k+2},\cdots,J_n$. Then
$L(\bar{T})\preceq_{s} L(\bar{S})$. Assume without loss of
generality that $L^{\bar{S}}_1\geq L^{\bar{S}}_2\geq\cdots\geq
L^{\bar{S}}_m$ and $L^{\bar{T}}_{\pi(1)}\geq
L^{\bar{T}}_{\pi(2)}\geq\cdots\geq L^{\bar{T}}_{\pi(m)}$, where
$\pi$ is a permutation of $\{1,2,\cdots,m\}$. For each $i$ with
$1\leq i\leq i_0$, the above claim implies that
$\sum^i_{j=1}L^S_j=\sum^i_{j=1}L^T_j\leq\sum^i_{j=1}L^{\bar{T}}_{\pi(j)}\leq\sum^i_{j=1}L^{\bar{S}}_j$.

Write $P=\sum^n_{j=1}p_j$, $Q=\sum^{i_0}_{i=1}L^S_i$ and
$\bar{Q}=\sum^{i_0}_{i=1}L^{\bar{S}}_i$. Then $Q\leq\bar{Q}$. Note
that, in the case $i_0=0$, we have $Q= \bar{Q} =0$. Let $J_d$ be the
last job scheduled on machine $M_{i_0+1}$ in $S$. By the choice of
$i_0$, $p_d\leq\frac{1}{2}(L^S_{i_0+1}-p_d)$. From the LPT
algorithm, we have $L^S_{i_0+1}-p_d\leq L^S_{j}$,
$j=i_0+1,i_0+2,\cdots,m$. Hence,
$$L^S_{i_0+1}\leq\frac{3}{2}(L^S_{i_0+1}-p_d)\leq\frac{3}{2} \cdot \frac{\sum^m_{j=i_0+1}L^S_{j}}{m-i_0}
=\frac{3}{2}\cdot\frac{1}{m-i_0}(P-Q).$$
 Thus, for each $i$ with  $i_0+1\leq
i\leq m$,  we have
\begin{equation}\label{eq1}
\sum^i_{j=1}L^S_j\leq Q+(i-i_0)L^{S}_{i_0+1}\leq
Q+\frac{3}{2}\cdot\frac{i-i_0}{m-i_0}(P-Q),
\end{equation}
and
\begin{equation}\label{eq2}
\sum^i_{j=1}L^{\bar{S}}_j\geq
\bar{Q}+(i-i_0)\frac{\sum^{i_0+1}_{j=m}L^{\bar{S}}_j}{m-i_0}=
\bar{Q}+\frac{i-i_0}{m-i_0}(P-\bar{Q})\geq
Q+\frac{i-i_0}{m-i_0}(P-Q).
\end{equation}
From (\ref{eq1}) and (\ref{eq2}), we conclude that
$\sum^i_{j=1}L^S_j\leq\frac{3}{2}\sum^i_{j=1}L^{\bar{S}}_j$.
Consequently, $s(S)\leq\frac{3}{2}$. It follows that $WAR(Pm(NP))
\leq \frac{3}{2}$ for $m\geq 4$.

Now let us consider problem $P3(NP)$. Let ${\cal I}$ be an instance.
Denote by $S$ the schedule which minimizes the makespan, and by $T$
the schedule which maximizes the machine cover. Without loss of
generality, we may assume that $L^S_1\geq L^S_2\geq L^S_3$,
$L^T_1\geq L^T_2\geq L^T_3$ and
$L^S_1+L^S_2+L^S_3=L^T_1+L^T_2+L^T_3=1$. Then $s(S)=
\frac{L^S_1+L^S_2}{L^T_1+L^T_2}$ and  $s(T) = \frac{L^T_1}{L^S_1}$.
Consequently, $s^*({\cal
I})\leq\min\{\frac{L^S_1+L^S_2}{L^T_1+L^T_2},\frac{L^T_1}{L^S_1}\}$.
Note that $L^T_1=1-L^T_2-L^T_3\leq 1-2L^T_3$ and
$L^S_1\geq\frac{L^S_1+L^S_2}{2}=\frac{1-L^S_3}{2}$. Then $s^*({\cal
I})\leq\min\{\frac{1-L^S_3}{1-L^T_3},\frac{1-2L^T_3}{\frac{1-L^S_3}{2}}\}$.
Set $x=1-2L^T_3$ and $t=1-L^S_3$. Then $\frac{2}{3}\leq t\leq 1$ and
$s^*({\cal I})\leq
 \min\{\frac{2t}{1+x},\frac{2x}{t}\}$.
If $x\geq\frac{\sqrt{1+4t^2}-1}{2}$, then $s^*({\cal
I})\leq\frac{2t}{1+x}\leq\frac{2t}{1+\frac{\sqrt{1+4t^2}-1}{2}}=\frac{\sqrt{1+4t^2}-1}{t}$.
If $x\leq\frac{\sqrt{1+4t^2}-1}{2}$, then $s^*({\cal
I})\leq\frac{2x}{t}\leq\frac{\sqrt{1+4t^2}-1}{t}$. Note that
$\frac{\sqrt{1+4t^2}-1}{t}\leq\sqrt{5}-1$ for all $t$ with
$\frac{2}{3}\leq t\leq 1$. It follows that $s^*({\cal
I})\leq\sqrt{5}-1$. The result follows.
\end{proof}

For problem $Pm(PP)$, \citet{MC59} presented an optimal algorithm
to generate a schedule which minimizes the makespan. A slight
modification of the algorithm can generate a schedule $S$ with
$s(S)=1$. \vspace{.3cm}

\textbf{Algorithm $MCR$ (with input $\mathcal{M}$ and
$\mathcal{J}$)}
\begin{itemize}
\item[1.] Finding the longest job $J_h$ in  $\mathcal{J}$. If
$p_h\leq\frac{\sum_{J_j\in\mathcal{J}}p_j}{|\mathcal{M}|}$, then
apply McNaughton's algorithm to assign all jobs in $\mathcal{J}$
to the machines in $\mathcal{M}$ evenly, and stop. Otherwise,
assign $J_h$ to an arbitrary machine $M_i\in\mathcal{M}$.
\item[2.] Reset $\mathcal{M}=\mathcal{M}\setminus\{M_i\}$ and
$\mathcal{J}=\mathcal{J}\setminus\{J_h\}$. If
$|\mathcal{J}|\neq0$, then go back to 1. Otherwise, stop.
\end{itemize}\vspace{.1cm}

\begin{lemma}\label{le:4}
Assume $p_1\geq p_2\geq\cdots\geq p_n$ and let $S$ be a preemptive
schedule with $L^S_1\geq L^S_2\geq\cdots\geq L^S_m$. Then
$\sum^k_{i=1}p_i\leq\sum^k_{i=1}L^S_i$, $k=1,2,\cdots,m$.
\end{lemma}
\begin{proof}
Let ${\cal J}_k=\{J_1,J_2,\cdots,J_k\}$. Then at most $k$ jobs in
${\cal J}_k$ can be processed simultaneously in the time interval
$[0,L^{S}_k]$ and at most $k-i$ jobs of ${\cal J}_k$ can be
processed simultaneously in the time interval
$[L^{S}_{k+1-i},L^{S}_{k-i}]$, $i=1,2,\cdots,k-1$. Therefore,
$\sum^k_{i=1}p_i\leq
kL^S_k+\sum^{k-1}_{i=1}(k-i)(L^{S}_{k-i}-L^{S}_{k+1-i})=\sum^k_{i=1}L^S_i$.
The lemma follows.
\end{proof}

\begin{theo}
$WAR(Pm(PP))=1$.
\end{theo}
\begin{proof}
Assume that $p_1\geq p_2\geq\cdots\geq p_n$. Let $i_0$ be the
largest job index such that
$p_{i}>\frac{\sum_{j=i_0}^{n}p_j}{m-i_0+1}$. If there is no such
index, we set $i_0=0$. Let $S$ be the preemptive schedule
generated by algorithm $MCR$ with $L^S_1\geq L^S_2\geq\cdots\geq
L^S_m$. Then we have
\begin{equation}\label{eq:6}
L^S_i=p_{i}, \; i=1,2,\cdots,i_0,
\end{equation}
and
\begin{equation}\label{eq:7}
L^S_i=\frac{\sum_{j=i_0+1}^{n}p_j}{m-i_0}, \;
i=i_0+1,i_0+2,\cdots,m.
\end{equation}
Let $T$ be a preemptive schedule with $L^T_1\geq L^T_2\geq\cdots\geq
L^T_m$. If $1\leq k\leq i_0$, by lemma \ref{le:4} and (\ref{eq:6}),
$\sum^k_{i=1}L^S_i=\sum^k_{i=1}p_i\leq\sum^k_{i=1}L^T_i$. If
$i_0+1\leq k\leq m$, by noting that
$\sum^{i_0}_{i=1}L^S_i\leq\sum^{i_0}_{i=1}L^T_i$, we have
$\sum^k_{i=1}L^S_i= \sum^{i_0}_{i=1}L^S_i+
\frac{k-i_0}{m-i_0}(\sum^n_{i=1}p_i-\sum^{i_0}_{i=1}L^S_i) \leq
\sum^{i_0}_{i=1}L^T_i+
\frac{k-i_0}{m-i_0}(\sum^n_{i=1}p_i-\sum^{i_0}_{i=1}L^T_i) \leq
\sum^k_{i=1}L^T_i$. Hence, $WAR(Pm(PP))=1$. The result follows.
\end{proof}

For problem $Pm(FP)$, the schedule $S$ averagely processing each
job on all machines clearly has $s(S)=1$. Then we have

\begin{theo}
$WAR(Pm(FP))=1$.
\end{theo}

\section{Related machines}
Assume that $s_1\geq s_2\geq\cdots\geq s_m$. We first present the
exact expression of $WAR(Qm(FP))$ on the  machine speeds $s_1, s_2,
\cdots, s_m$.  Then we show that it is a lower bound for
$WAR(Qm(PP))$ and $WAR(Qm(NP))$.

The fractional processing mode means that all jobs can be merged
into a single job with processing time equal to the sum of
processing times of all jobs. Thus we may assume that ${\cal I}$ is
an instance of $Qm(FP)$ with just one job $J_{\cal I}$. Suppose
without loss of generality that $p_{\cal I}=1$. A schedule $S$  of
${\cal I}$ is called \emph{regular} if $L^{S}_1\geq
L^{S}_2\geq\cdots\geq L^{S}_m$. Then $\overleftarrow{L(S)}=L(S)$ if
$S$ is regular. The following lemma can be observed from the basic
mathematical knowledge.

\begin{lemma}\label{le:11}
Suppose that $x_1\geq{x_2}\geq\cdots\geq{x_n}\geq0$ and
$y_1\geq{y_2}\geq\cdots\geq{y_n}\geq0$. Then
$\sum_{i=1}^{n}x_iy_{\pi(i)}\leq\sum_{i=1}^{n}x_iy_i$ for any
permutation $\pi$ of $\{1,2,\cdots,n\}$.
\end{lemma}

\begin{lemma}\label{le:12}
For any schedule $T$ of $\mathcal{I}$, there exists a regular
schedule $S$ such that $L(S)\preceq_{c}\overleftarrow{L(T)}$.
\end{lemma}
\begin{proof}
Let $T$ be a schedule of $\mathcal{I}$ and $\pi$  a permutation of
$\{1,2,\cdots,m\}$ such that $L^T_{\pi(1)}\geq
L^T_{\pi(2)}\geq\cdots\geq L^T_{\pi(m)}$. By lemma \ref{le:11},
$\sum^m_{i=1}s_iL^{T}_{\pi(i)}\geq\sum^m_{i=1}s_{\pi(i)}L^{T}_{\pi(i)}\geq
1$.  Let $i_0$ be the smallest machine index such that
$\sum^{i_0}_{i=1}s_iL^{T}_{\pi(i)}\geq1$. Let $S$ be the schedule in
which a part of processing time $s_iL^{T}_{\pi(i)}$ is assigned to
$M_i$, $i=1,2,\cdots,i_0-1$, and the rest part of processing time
$1-\sum^{i_0-1}_{i=1}s_iL^{T}_{\pi(i)}$ is assigned to $M_{i_0}$.
Then we have $L^S_i=L^T_{\pi(i)}$, for $i=1,2,\cdots,i_0-1$,
$L^S_{i_0}=\frac{1-\sum^{i_0-1}_{i=1}s_iL^{T}_{\pi(i)}}{s_{i_0}}
\leq\frac{\sum^{i_0}_{i=1}s_iL^{T}_{\pi(i)}-\sum^{i_0-1}_{i=1}s_iL^{T}_{\pi(i)}}{s_{i_0}}
=L^T_{\pi(i_0)}$, and $L^S_i=0\leq L^T_{\pi(i)}$ for
$i=i_0+1,i_0+2,\cdots,m$. It can be observed that $S$ is regular and
$L(S)\preceq_{c}\overleftarrow{L(T)}$. The lemma follows.
\end{proof}

Let $f(i)$ be the infimum of the sum of the first $i$ coordinates of
$\overleftarrow{L(T)}$ in all feasible schedule $T$ of
$\mathcal{I}$, $i=1,2,\cdots,m$. By lemma \ref{le:12}, we have
$f(i)=\inf\{\sum_{k=1}^{i}L^{S}_{k}:  S\mbox{ is
regular}\},i=1,2,\cdots,m$. Then, for each schedule $T$ of
$\mathcal{I}$ with $L^T_{\pi(1)}\geq L^T_{\pi(2)}\geq\cdots\geq
L^T_{\pi(m)}$ for some permutation $\pi$ of $\{1,2,\cdots,m\}$, we
have
\begin{equation}\label{eq:1-1}
s(T)=\max_{1\leq i\leq
m}\left\{\frac{\sum_{k=1}^{i}L^{\tau}_{\pi(k)}}{f(i)}\right\}.
\end{equation}
The following lemma gives the exact expression for each $f(i)$.

\begin{lemma}\label{le:13}
$f(i)=\left\{\begin{array}{cc}
               \frac{i}{\sum^m_{k=1}s_k}, & i\leq\frac{\sum^m_{k=1}s_k}{s_1}
               ;\\[0.2cm]
               \frac{1}{s_1}, & i>\frac{\sum^m_{k=1}s_k}{s_1}.
             \end{array}
\right.$
\end{lemma}
\begin{proof}
Fix index $i$ and let $S$ be a regular schedule. Then we have
\begin{equation}\label{eq:10}
L^{S}_1\geq L^{S}_2\geq\cdots\geq L^{S}_m
\end{equation}
and
\begin{equation}\label{eq:11}
\sum^m_{i=1}s_iL^{S}_i\geq1.
\end{equation}
So we only need to find a regular schedule $S$ meeting (\ref{eq:10})
and (\ref{eq:11}) such that $\sum_{k=1}^{i}L^{S}_{k}$ reaches the
minimum.

If $i\leq\frac{\sum^m_{k=1}s_k}{s_1}$, by (\ref{eq:10}) and
(\ref{eq:11}),
\begin{eqnarray}
\sum_{t=1}^{i}\left(\frac{\sum^m_{k=1}s_k}{i}\right)L^{S}_{t}&=&\sum^i_{t=1}s_tL^{S}_t+\sum_{t=1}^{i}\left(\frac{\sum^m_{k=1}s_k}{i}-s_t\right)L^{S}_{t}\nonumber\\
&\geq&\sum^i_{t=1}s_tL^{S}_t+\sum_{t=1}^{i}\left(\frac{\sum^m_{k=1}s_k}{i}-s_t\right)L^{S}_{i+1}\nonumber\\
&=&\sum^i_{t=1}s_tL^{S}_t+\left(\sum_{t=i+1}^{m}s_t\right)L^{S}_{i+1}\nonumber\\
&\geq&\sum^i_{t=1}s_tL^{S}_t+\sum^m_{t=i+1}s_tL^{S}_t=\sum^m_{t=1}s_tL^{S}_t\geq1.\nonumber
\end{eqnarray}
The equality holds if and only if
$L^S_1=L^S_2=\cdots=L^S_m=\frac{1}{\sum^m_{k=1}s_k}$. Then the
regular schedule $S$ can be defined by the way that  a part of
processing time $\frac{s_k}{\sum^m_{k=1}s_k}$ is assigned to $M_k$,
$k=1,2,\cdots,m$. Thus, $f(i)=\frac{i}{\sum^m_{k=1}s_k}$.

If $i>\frac{\sum^m_{k=1}s_k}{s_1}$,   we can similarly deduce
\begin{eqnarray}
\sum_{k=1}^{i}s_1L^{S}_{k}&=&\sum_{k=1}^{i}s_kL^{S}_{k}+\sum_{k=1}^{i}(s_1-s_k)L^{S}_{k}\nonumber\\
&\geq&\sum_{k=1}^{i}s_kL^{S}_{k}+\sum_{k=1}^{i}(s_1-s_k)L^{S}_{i}\nonumber\\
&=&\sum_{k=1}^{i}s_kL^{S}_{k}+\left(is_1-\sum^i_{k=1}s_k\right)L^{S}_{i}\nonumber\\
&\geq&\sum_{k=1}^{i}s_kL^{S}_{k}+\left(\sum^m_{k=1}s_k-\sum^i_{k=1}s_k\right)L^{S}_{i}\nonumber\\
&\geq&\sum_{k=1}^{i}s_kL^{S}_{k}+\sum^m_{k=i+1}s_kL^{S}_{k}=\sum^m_{k=1}s_kL^{S}_k\geq1.\nonumber
\end{eqnarray}
The equality holds if and only if
$L^S_1=\frac{1}{s_1},L^S_2=\cdots=L^S_m=0$. Then the regular
schedule $S$ can be defined by the way that $J_{\mathcal{I}}$ is
scheduled totally on $M_1$ in $S$. Thus $f(i)=\frac{1}{s_1}$. The
lemma follows.
\end{proof}

By lemma \ref{le:12}, $s^*(\mathcal{I})=\inf\{s(S): S\mbox{ is
regular}\}$. For each regular schedule $S$, by (\ref{eq:1-1}) and
lemma \ref{le:13}, we have $\sum^i_{k=1}L^{S}_k\leq{s(L(S))}{f(i)}$
for $i=1,2,\cdots,m.$.

Let $s_{m+1}=0$ and $\frac{\sum^m_{i=1}s_i}{s_1}=t+\Delta$, where
$t$ with  $1\leq t\leq m$ is a positive integer and
$0\leq\Delta<1$. By lemma \ref{le:13}, we have
\begin{equation}\label{eq:13}
i\cdot\frac{s(L(S))}{\sum^m_{k=1}s_k}\geq\sum^i_{k=1}L^{S}_k, \;
i=1,2,\cdots,t.
\end{equation}
and
\begin{equation}\label{eq:14}
\frac{s(L(S))}{s_1}\geq\sum^i_{k=1}L^{S}_k, \; i=t+1,t+2,\cdots,m.
\end{equation}
From (\ref{eq:13}) and (\ref{eq:14}), we have
$\sum^t_{i=1}(s_i-s_{i+1})\cdot
i\cdot\frac{s(L(S))}{\sum^m_{i=1}s_i}+\sum^m_{i=t+1}(s_i-s_{i+1})\frac{s(L(S))}{s_1}
\geq \sum^t_{i=1}(s_i-s_{i+1})\sum^i_{t=1}L^{S}_t
+\sum^m_{i=t+1}(s_i-s_{i+1})\sum^i_{t=1}L^{S}_t =
\sum^m_{i=1}s_iL^{S}_i=1$. Hence,
$s(S)\geq\frac{\sum^m_{i=1}s_i}{\sum^t_{i=1}s_i+\left(\frac{\sum^m_{i=1}s_i}{s_1}-t\right)s_{t+1}}=\frac{\sum^m_{i=1}s_i}{\sum^t_{i=1}s_i+\Delta
s_{t+1}}$. Note that the equality holds if and only if
$L^S_1=L^S_2=\cdots=L^S_t=\frac{1}{\sum^t_{i=1}s_i+\Delta s_{t+1}}$,
$L^S_{t+1}=\frac{\Delta}{\sum^t_{i=1}s_i+\Delta s_{t+1}}$ and
$L^S_{t+2}=L^S_{t+3}=\cdots=L^S_m=0$. Then the corresponding regular
schedule $S$ can be defined by the way that a part of processing
time $\frac{s_i}{\sum^t_{k=1}s_k+\Delta s_{t+1}}$ is assigned to
$M_i$, $i=1,2,\cdots,t$, and the rest part of processing time
$\frac{\Delta s_{t+1}}{\sum^t_{i=1}s_i+\Delta s_{t+1}}$ is assigned
to $M_{t+1}$. Hence,
$s^*({\mathcal{I}})=\frac{\sum^m_{i=1}s_i}{\sum^t_{i=1}s_i+\Delta
s_{t+1}}$. Consequently, $WAR(Qm(FP)) =
\frac{\sum^m_{i=1}s_i}{\sum^t_{i=1}s_i+\Delta s_{t+1}}$  if the
machine speeds are fixed.

If the machine speeds are parts of the input, by the fact that
$s_1\geq s_2\geq\cdots\geq s_m$, we have
\begin{equation}\label{eq:15}
\frac{\sum^t_{i=2}s_i+\Delta
s_{t+1}}{t-1+\Delta}\geq\frac{\sum^m_{i=2}s_i}{m-1}.
\end{equation}
Let $\theta=\frac{\sum^m_{i=2}s_i}{m-1}$ and
$\vartheta=\frac{s_1}{\theta}>1$. Then
\begin{equation}\label{eq:16}
t+\Delta=\frac{\sum^m_{i=1}s_i}{s_1}=\frac{s_1+(m-1)\theta}{s_1}=\frac{\vartheta+m-1}{\vartheta}.
\end{equation}
Obviously,
$\frac{m}{\vartheta-1}+(\vartheta-1)\geq2\sqrt{\frac{m}{\vartheta-1}(\vartheta-1)}=2\sqrt{m}$.
By (\ref{eq:15}) and (\ref{eq:16}), we have
$\frac{\sum^m_{i=1}s_i}{\sum^t_{i=1}s_i+\Delta s_{t+1}} =
\frac{s_1+(m-1)\frac{\sum^m_{i=2}s_i}{m-1}}{s_1+(t-1+\Delta)\frac{\sum^t_{i=2}s_i+\Delta
s_{t+1}}{t-1+\Delta}} \leq
\frac{s_1+(m-1)\frac{\sum^m_{i=2}s_i}{m-1}}{s_1+(t-1+\Delta)\frac{\sum^m_{i=2}s_i}{m-1}}
= 1+\frac{m-1}{(\frac{m}{\vartheta-1}+(\vartheta-1))+2} \leq
1+\frac{m-1}{2\sqrt{m}+2}=\frac{\sqrt{m}+1}{2}$. So we have
$s^*(\mathcal{I})\leq\frac{\sqrt{m}+1}{2}$ and therefore
$WAR(Qm(FP)) \leq\frac{\sqrt{m}+1}{2}$.

To show that $WAR(Qm(FP)) = \frac{\sqrt{m}+1}{2}$, we consider the
following instance $\mathcal{I}$ with $p_\mathcal{I} =1$,
$s_1=s=\sqrt{m}+1>1$ and $s_2=s_3=\cdots=s_m=1$. Let $S$ be a
regular schedule and write   $x=sL^{S}_1$. Then
$\sum^m_{t=2}L^{S}_t=1-x$. By lemma \ref{le:13} and (\ref{eq:1-1}),
we have
$s(S)\geq\max\left\{\frac{L^{S}_1}{f(1)},\frac{\sum_{i=1}^{m}L^{S}_i}{f(m)}\right\}
=  \max\left\{\frac{x(s+m-1)}{s},x+s(1-x)\right\} \geq
\frac{s^2+sm-s}{s^2+m-1}=\frac{\sqrt{m}+1}{2}$, where the inequality
follows from the fact that  $\frac{x(s+m-1)}{s}$ is an increasing
function in $x$ while $x+s(1-x)$ is a decreasing function in $x$ and
they meet with $\frac{s^2+sm-s}{s^2+m-1}$ when
$x=\frac{s^2}{s^2+m-1}$. Then
$s^*(\mathcal{I})\geq\frac{\sqrt{m}+1}{2}$. Consequently,
$WAR(Qm(FP)) = \frac{\sqrt{m}+1}{2}$.

The above discussion leads to the following conclusion.

\begin{theo}\label{th:14}
If the machine speeds $s_1,s_2,\cdots,s_m$ are fixed, then
$WAR(Qm(FP) = \frac{\sum^m_{i=1}s_i}{\sum^t_{i=1}s_i+\Delta
s_{t+1}}$, where $\frac{\sum^m_{i=1}s_i}{s_1}=t+\Delta$, $1\leq
t\leq m$ is a positive integer and $0\leq\Delta<1$. If the machine
speeds $s_1,s_2,\cdots,s_m$ are parts of the input, then
$WAR(Qm(FP) = \frac{\sqrt{m}+1}{2}$.
\end{theo}

\begin{lemma}\label{le:15}
If the machine speeds $s_1,s_2,\cdots,s_m$ are fixed, then
$WAR(Qm(NP)) \geq WAR(Qm(FP))$ and  $WAR(Qm(PP)) \geq WAR(Qm(FP))$.
\end{lemma}
\begin{proof}
We only consider the non-preemptive processing mode. For the
preemptive processing mode, the result can be similarly proved.
Given a schedule $S$, we denote by $\pi^{S}$ the permutation of
$\{1,2,\cdots,m\}$ such that $L^S_{\pi^{S}(1)}\geq
L^S_{\pi^{S}(2)}\geq\cdots\geq L^S_{\pi^{S}(m)}$.

Suppose without loss of generality that $s_m=1$. Write $\eta
=WAR(Qm(NP))$. Let $\mathcal{I}$ be an instance of $Q_m(FP)$ with
only one job $J_{\mathcal{I}}$ of processing time 1. For each $i$,
set $f(i)$ to be the infimum of $\sum_{k=1}^{i}L^{S}_{\pi^{S}(k)}$
of schedule $S$ over all fractional schedules of $\mathcal{I}$. We
only need to show that $s^*(\mathcal{I})\leq\eta$.

Assume to the contrary that $s^*(\mathcal{I})>\eta$. Let
$\epsilon>0$ be a sufficiently small number such that
$\eta(f(i)+i\epsilon)<s^*(\mathcal{I})f(i)$, $i=1,2,\cdots,m$. Let
$\mathcal{H}$ be an instance of $Q_m(NP)$ such that the total
processing time of jobs is equal to $1$ and the processing time of
each job is at most $\epsilon$. For each $i$, let $g(i)$ be the
infimum of $\sum_{k=1}^{i}L^{S}_{\pi^{S}(k)}$ of schedule $S$ over
all feasible schedules of $\mathcal{H}$. We assert that
\begin{equation}\label{eq:17}
g(i)\leq f(i)+i\epsilon, \; i=1,2,\cdots,m.
\end{equation}

To the end, let $S_i$ be the regular schedule of $\mathcal{I}$ such
that $\sum^{i}_{k=1}L^{S_i}_{k}=f(i)$, $i=1,2,\cdots,m$. Fix index
$i$, we construct a non-preemptive schedule $S$ of $\mathcal{H}$
such that $\sum_{k=1}^{i}L^{S}_{\pi^{S}(k)}\leq f(i)+i\epsilon$.
This leads to $g(i)\leq\sum_{k=1}^{i}L^{S}_{\pi^{S}(k)}\leq
f(i)+i\epsilon$, and therefore, proves the assertion. The
construction of $S$ is stated as follows. First, we assign jobs to
$M_i$ one by one until $L^S_1\geq L^{S_{i}}_1$. Then we assign the
rest jobs to $M_2$ one by one until $L^S_2\geq L^{S_i}_2$. This
procedure is repeated  until all jobs are assigned. According to the
construction of $S$, we have $L^S_{k}\leq
L^{S_{i}}_k+\frac{\epsilon}{s_k}\leq L^{S_{i}}_k+\epsilon$,
$k=1,2,\cdots,m$. Note that $L^{S_i}_1\geq L^{S_i}_2\geq\cdots\geq
L^{S_i}_m$. Then
$\sum_{k=1}^{i}L^{S}_{\pi^{S}(k)}\leq\sum_{k=1}^{i}(L^{S_i}_{\pi^{S}(k)}+\epsilon)
\leq\sum^{i}_{k=1}L^{S_i}_{k}+i\epsilon=f(i)+i\epsilon$.

Let $R$ be the schedule of $\mathcal{H}$ such that
$s(R)=s^*(\mathcal{H})$. It can be observed that there exists a
schedule $T$ of $\mathcal{I}$ such that $L(T)\preceq_{c}L(R)$.
Hence, for each $i$ with $1\leq i\leq m$, we have
$\sum^{i}_{k=1}L^T_{\pi^T(k)} \leq \sum^{i}_{k=1}L^{R}_{\pi^T(k)}
\leq\sum^{i}_{k=1}L^{R}_{\pi^R(k)} \leq {s(R)g(i)} \leq
{s^*(\mathcal{H})(f(i)+i\epsilon)} \leq
\eta(f(i)+i\epsilon)<s^*(\mathcal{I})f(i)$. This contradicts the
definition of $s^*(\mathcal{I})$. So $s^*(\mathcal{I})\leq\eta$. The
result follows.
\end{proof}

By theorem \ref{th:14} and lemma \ref{le:15}, the following theorem
holds.

\begin{theo}\label{th:16}
If the machine speeds $s_1,s_2,\cdots,s_m$ are fixed, then
$WAR({\cal P}) \geq \frac{\sum^m_{i=1}s_i}{\sum^t_{i=1}s_i+\Delta
s_{t+1}}$ for ${\cal P} \in \{ Qm(NP), Qm(PP)\}$, where
$\frac{\sum^m_{i=1}s_i}{s_1}=t+\Delta$, $t$ is a positive integer
with  $1\leq t\leq m$, and $0\leq\Delta<1$. If the machine speeds
$s_1,s_2,\cdots,s_m$  are parts of the input, then $WAR({\cal P})
\geq \frac{\sqrt{m}+1}{2}$ for ${\cal P} \in \{ Qm(NP), Qm(PP)\}$.
\end{theo}

\section{Unrelated machines}

Since $Qm$ is a special version of $Rm$, from the results in the
previous section,   the weak simultaneous approximation ratio is at
least $\frac{\sqrt{m}+1}{2}$ for each of $Rm(NP)$, $Rm(PP)$ and
$Rm(FP)$.   The following lemma establishes an upper bound of the
weak simultaneous approximation ratio for the three problems.

\begin{lemma}\label{le:18}
$WAR({\cal P})  \leq \sqrt{m}$ for ${\cal P} \in \{ Rm(NP),
Rm(PP), Rm(FP)\}$.
\end{lemma}
\begin{proof}
Let ${\cal I}$ be an instance of $R_m(NP)$, $R_m(PP)$ or $R_m(FP)$.
Let $S$ be a schedule which minimizes the makespan with $L^S_1\geq
L^S_2\geq\cdots\geq L^S_m$. Write $p_{[j]}=\min_{1\leq i\leq
m}\{p_{ij}\}$.

If $L^S_1\leq\frac{\sum^n_{j=1}p_{[j]}}{\sqrt{m}}$, let $T$ be a
feasible schedule with  $L^T_{\pi(1)}\geq
L^T_{\pi(2)}\geq\cdots\geq L^T_{\pi(m)}$ for some permutation
$\pi$ of $\{1,2,\cdots,m\}$. For each $i$, we have
$\sum^i_{k=1}L^S_k\leq
iL^S_1\leq\sqrt{m}\cdot\frac{i}{m}\sum^n_{j=1}p_{[j]}
\leq\sqrt{m}\sum^i_{k=1}L^T_{\pi(k)}$. This means that
$s^{*}({\cal I})\leq\sqrt{m}$.

If $L^S_1>\frac{\sum^n_{j=1}p_{[j]}}{\sqrt{m}}$, let $R$ be  the
schedule in which each job $J_j$  is assigned to the machine $M_i$
with $p_{ij}= p_{[j]}$. Let $O$ be an arbitrarily feasible schedule,
and let ${\pi}_1$ and ${\pi}_2$ be two permutations of
$\{1,2,\cdots,m\}$ such that $L^R_{{\pi}_1(1)}\geq
L^R_{{\pi}_1(2)}\geq\cdots\geq L^R_{{\pi}_1(m)}$ and
$L^O_{{\pi}_2(1)}\geq L^O_{{\pi}_2(2)}\geq\cdots\geq
L^O_{{\pi}_2(m)}$. For each $i$, we have
$\sum^i_{k=1}L^R_{{\pi}_1(k)}
\leq\sum^m_{k=1}L^R_{{\pi}_1(k)}=\sum^n_{j=1}p_{[j]} <\sqrt{m}L^S_1
\leq\sqrt{m}L^O_{{\pi}_2(1)}
\leq\sqrt{m}\sum^i_{k=1}L^{O}_{{\pi}_2(k)}$. This also means that
$s^{*}({\cal I})\leq\sqrt{m}$. The lemma follows.
\end{proof}

Combining with the results of the previous section, we have the
following theorem.

\begin{theo}
For each problem ${\cal P} \in \{Qm(NP), Qm(PP), Qm(FP), Rm(NP),
Rm(PP), Rm(FP)\}$, we have $\frac{\sqrt{m}+1}{2} \leq WAR({\cal
P})  \leq \sqrt{m}$.
\end{theo}

\section*{Acknowledgments}
The authors would like to thank the associate editor and two
anonymous referees for their constructive comments and kind
suggestions.

\bibliographystyle{plainnat}

\end{document}